\documentclass[journal]{IEEEtran}

\usepackage{algorithm}
\usepackage{algorithmic}
\usepackage{amssymb}
\usepackage{amsfonts}
\usepackage{amsmath}
\usepackage{mathtools}
\usepackage{amsthm}
\usepackage{bm}
\usepackage{booktabs}
\usepackage{cases}
\usepackage{cite}
\usepackage{color}
\usepackage{comment}
\usepackage{empheq}
\usepackage{enumerate}
\usepackage{enumitem}
\usepackage{framed}
\usepackage{graphicx}
\usepackage{latexsym}
\usepackage{makecell}
\usepackage{matlab-prettifier}
\usepackage{tikz}
\usepackage{url}
\usepackage{blkarray}
\usepackage[hidelinks]{hyperref}
\usepackage{newtxtext}
\usepackage{siunitx}
\newtheorem{theorem}{Theorem}

\newtheorem{remark}{Remark}
\newtheorem{lemma}{Lemma}
\newtheorem{corollary}{Corollary}

\newcommand{\tred}{\textcolor{black}}
\newcommand{\tblack}{\textcolor{black}}

\mathtoolsset{showonlyrefs=true}  

\DeclareMathOperator{\diag}{diag}

\newcommand{\R}{\mathbb{R}}
\newcommand{\C}{\mathbb{C}}

\newcommand{\conv}{\mathbin{\ast}}

\begin{document}
\title{A Deep State-Space Model Compression Method using Upper Bound on Output Error}
\author{Hiroki Sakamoto and Kazuhiro Sato\thanks{H. Sakamoto and K. Sato are with the Department of Mathematical Informatics, Graduate School of Information Science and Technology, The University of Tokyo, Tokyo 113-8656, Japan, email: soccer-books0329@g.ecc.u-tokyo.ac.jp (H. Sakamoto), kazuhiro@mist.i.u-tokyo.ac.jp (K. Sato) }}

\maketitle
\thispagestyle{empty}
\pagestyle{empty}

\begin{abstract}
We study deep state-space models (Deep SSMs) that contain linear quadratic-output (LQO) systems as internal blocks and present a compression method with a provable output error guarantee. 
We first derive an upper bound on the output error between two Deep SSMs and show that the bound can be expressed in terms of the $h^2$-error norms between the layerwise LQO systems. 
\tblack{In particular, we show that reducing the $h^2$ approximation errors of the LQO systems placed in shallow layers is effective in reducing the derived upper bound on the output error.}
Next, we formulate an optimization problem for the derived upper bound and develop a gradient-based MOR method. 
In the numerical experiments, using the IMDb task from the LRA benchmark, we demonstrate the effectiveness of the proposed upper-bound-based compression method. 
In particular, \tblack{we show that the number of trainable parameters can be reduced by approximately 60\% without retraining while maintaining the performance of the original model.}
\end{abstract}


\section{Introduction} \label{sec:intro}
Deep state-space models (Deep SSMs)~\cite{gu2022efficiently,gu2022parameterization, smith2023simplified, gu2023mamba} are deep models that incorporate linear state-space models (SSMs) into intermediate layers, and they have attracted attention as sequence models that can efficiently handle long-range dependencies and nonlinearities.
As representative examples, S5~\cite{smith2023simplified} and Mamba~\cite{gu2023mamba} achieve state-of-the-art performance on various tasks in the Long Range Arena (LRA) benchmark~\cite{tay2021long} without using the attention mechanism employed in Transformers~\cite{vaswani2017attention}, and their effectiveness has been reported on real-world data~\cite{wang2024state}.
In theory, Deep SSMs combined with simple nonlinear layers have been shown to possess Transformer-comparable capabilities in dynamic token selection and certain nonparametric regression settings~\cite{nishikawa2025state}.
Achieving high performance often requires a sufficient number of parameters; in particular, for Deep SSMs, one can build high-performance models by enlarging the parameter size of the embedded linear SSMs.

When deploying trained models to a variety of tasks, it is desirable to obtain compact models with fewer parameters while maintaining the performance of large models.
To construct small-scale Deep SSMs, one may apply general model-compression techniques developed for deep learning models.
Well-known approaches include pruning, quantization, and knowledge distillation, which are effective also for compressing Deep SSMs~\cite{cheng2018model}.
On the other hand, it has recently been recognized that model order reduction (MOR), a classical topic in systems and control theory, is effective as a compression method for Deep SSMs~\cite{ezoe2024model, forgione2024model, gwak2024layer, sakamoto2025compression,chahine2026curious}.

MOR-based compression seeks to reduce the number of parameters of the linear SSMs placed within Deep SSMs while preserving the performance of the original large model.
\cite{ezoe2024model} and \cite{sakamoto2025compression} used Balanced Truncation (BT)- and $\mathcal{H}^2$-based MOR, respectively, as initialization for retraining compressed models.
\cite{forgione2024model} introduced MOR-oriented regularization into the training objective and applied BT to obtain one-shot compressed models.
\cite{gwak2024layer} proposed a post-training state-pruning method based on an $\mathcal{H}^{\infty}$-based indicator, while \cite{chahine2026curious} performed MOR during training to construct compressed SSMs.

As shown in Fig.~\ref{fig:comparison}, in these existing studies, compression is performed by applying MOR individually to each linear SSM placed at intermediate layers, based on various indicators and training schemes.
While this approach constructs reduced models that approximate the input-output behavior of each individual linear SSM, it can fail to sufficiently approximate the final output of the original Deep SSM as a whole.

In this work, to adequately approximate the overall performance of the original Deep SSM, we construct reduced models that reflect inter-layer interactions (see Fig.~\ref{fig:comparison}).
\tred{In particular, to the best of our knowledge, this is the first study from the viewpoint of systems and control theory to propose a compression method that directly guarantees the overall output performance of a Deep SSM.}
To facilitate theoretical analysis based on MOR theory~\cite{benner2021gramians,reiter2025h,zulfiqar2024time}, \tblack{we consider Deep SSMs that internally contain linear quadratic-output (LQO) systems, which can be interpreted as a generalization of the S5 architecture~\cite{smith2023simplified} that contains a quadratic activation function immediately after each linear state-space block,} and we develop a compression method with a provable upper bound on the output error.

The contributions of this paper are:
\begin{itemize}
  \item To construct a reduced Deep SSM that minimizes the output error~$\|s_{\mathrm{out}}-\hat{s}_{\mathrm{out}}\|_{\ell^{\infty}_L}$ for any input sequence~$s_{\mathrm{in}}$, we derive an upper bound on the output error between two Deep SSMs, as illustrated in Fig.~\ref{fig:comparison} (top), because directly optimizing this quantity is challenging.
  \tblack{In particular, we show that reducing the $h^2$ approximation errors of the LQO systems placed in shallow layers contributes to reducing the derived upper bound on the output error.}
  \tblack{Furthermore, the derived bound shows that minimizing the $h^2$ error of each individual SSM within a Deep SSM, as shown in Fig.~\ref{fig:comparison} (bottom), leads to minimizing an upper bound of the overall output error.
  This result provides a theoretical basis for existing layerwise MOR-based compression for Deep SSMs~\cite{ezoe2024model,forgione2024model,gwak2024layer,sakamoto2025compression,chahine2026curious}.}
  Finally, we introduce an optimization algorithm with a stationary point guarantee that minimizes the bound while preserving the unique properties of Deep SSMs.
  \item We show numerically that the proposed upper-bound-based compression constructs an accurate compression model.
  \tblack{In particular, we show that, by using the design principle for MOR-based compression, namely, assigning larger reduced state dimensions to shallow layers to reduce the derived upper bound, a compressed model with high accuracy can be constructed without retraining (in a one-shot manner).}
  Moreover, for the retraining-free compression methods~\cite{gwak2024layer}, when compared in terms of accuracy at the same state dimension (cf.~\cite[Fig.2]{gwak2024layer}), the proposed method, despite differences in the Deep SSM architecture, constructs compressed models with superior performance.
  \tred{This provides a practical contribution: it enables low-cost, high-performance deployment in settings where retraining is impractical or impossible, e.g., under strict compute constraints.}
\end{itemize}

The paper is organized as follows: Section~\ref{sec:system} introduces discrete-time complex LQO systems; Section~\ref{sec:problem_setting} formulates the problem; Section~\ref{sec:error_analysis} derives the output error bound; Section~\ref{sec:proposed_mor} presents the compression strategy and optimization algorithm; Section~\ref{sec:experiment} reports experiments; and Section~\ref{sec:summary} concludes.

\paragraph*{Notation}
For vectors, $\|\cdot\|$ denotes the Euclidean norm and $\|\cdot\|_\infty$ denotes the infinity norm.  
For matrices, $\|\cdot\|_2$ denotes the operator (spectral) norm and $\|\cdot\|$ denotes the Frobenius norm.  
Let $(X,\|\cdot\|)$ be a normed vector space (e.g., $X=\R^m$ or $\C^m$ with the Euclidean norm) and define, for $1\le p\le\infty$, $\ell^p(X):=\left\{x=(x_k)_{k\ge0}\subset X \mid \sum_{k\geq 0}\|x_k\|^{p}< \infty\right\}$, and $\ell^{\infty}(X):=\left\{x=(x_k)_{k\ge0}\subset X \mid \sup_{k\geq 0}\|x_k\|< \infty\right\}$.
Throughout, we write $\ell^p:=\ell^p(X)$ for $1\le p\le\infty$.
For \(x\in \ell^{p}\), define $\|x\|_{\ell^{p}}:=\big(\sum_{k\ge0}\|x_k\|^p\big)^{1/p}$ for $1\le p<\infty$; 
for a finite horizon $L$, define 
$\|x\|_{\ell^{p}_{L}}:=\big(\sum_{k=0}^{L-1}\|x_k\|^p\big)^{1/p}$ for $1\le p<\infty$ and 
$\|x\|_{\ell^{\infty}_{L}}:=\max_{0\le k\le L-1}\|x_k\|$.

\tred{The overline $\overline{(\cdot)}$ denotes complex conjugation.}
The conjugate transpose is $(\cdot)^{*}$ (so $A=A^{*}$ means Hermitian; for real matrices, symmetric), and $(\cdot)^{\top}$ denotes transpose.
For a matrix $X$, $\mathrm{vec}(X)$ stacks its columns into a single vector.
The symbol $\otimes$ denotes the Kronecker product, and $\odot$ denotes the elementwise (Hadamard) product.

For a real-valued function $f$, $\mathrm{d}f$ denotes its (total) differential; for a perturbation $\mathrm{d}\theta$ of $\theta$, we write $\mathrm{d}f=\langle\nabla_{\theta}f,\mathrm{d}\theta\rangle$, where the inner product is $\langle A,B\rangle:=\mathrm{Re}\big(\mathrm{tr}(A^{*}B)\big)$.
For multiple variables $\theta=(\theta_1,\theta_2,\dots)$, we write $\mathrm{d}f=\sum_j\langle\nabla_{\theta_j}f,\mathrm{d}\theta_j\rangle$.

The symbol $\conv$ denotes discrete-time convolution, defined for $k\ge0$ by $(x\conv h)[k]:=\sum_{m=0}^{k} x[m]h[k-m]$. 
For $1\le p,q,r\le\infty$ satisfying $\tfrac{1}{p}+\tfrac{1}{q}=1+\tfrac{1}{r}$, $x\in \ell^{p}$ and $h\in \ell^{q}$ obey the Young's convolution inequality $\|x\conv h\|_{\ell^{r}}\ \le\ \|x\|_{\ell^{p}}\;\|h\|_{\ell^{q}}$.
For the finite-horizon norms $\ell^{p}_{L}$, the same bounds hold.

\begin{figure}[htbp]
 \centering
 \includegraphics[width=\linewidth]{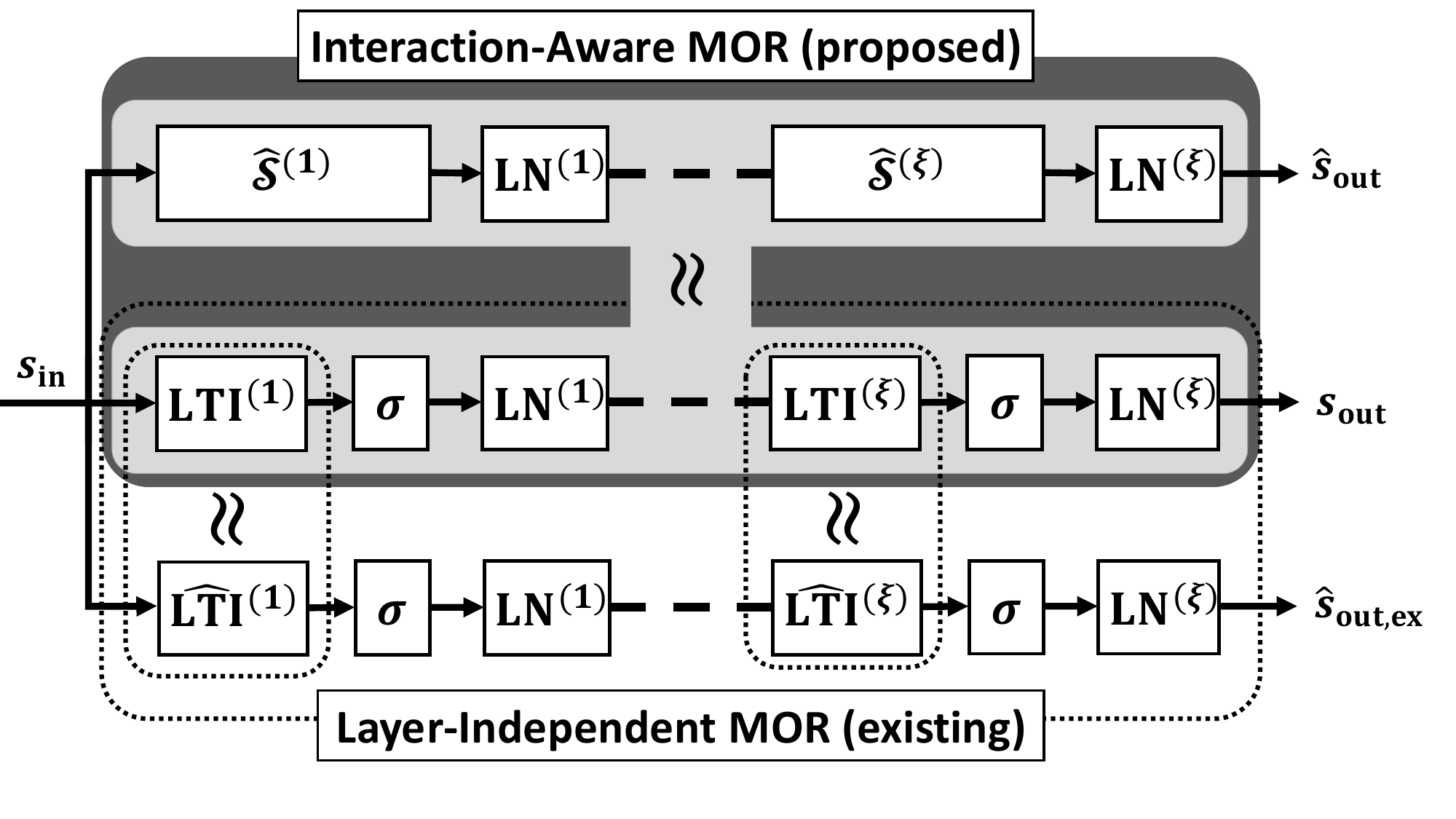}
 \caption{
 Comparison between existing MOR methods and the proposed method for $\xi$-layer Deep SSMs.
 Existing methods perform MOR independently on the linear time-invariant (LTI) subsystems within a Deep SSM and produce \(\hat{s}_{\mathrm{out,ex}}\).
 In contrast, the proposed method constructs a reduced Deep SSM with reduced LQO systems $\{\hat{\mathcal{S}}^{(i)}\}_{i=1}^{\xi}$ that approximates the pretrained Deep SSM’s output \(s_{\mathrm{out}}\) while accounting for inter-layer interactions, \tblack{the quadratic output maps \(\sigma\) of the LQO systems}, and LayerNorm (LN); \tred{in particular, it aims to reduce \(\|s_{\mathrm{out}}-\hat{s}_{\mathrm{out}}\|_{\ell_L^\infty}\) by minimizing its derived upper bound.}
 \label{fig:comparison}}
\end{figure}


\section{Discrete-time complex LQO systems}\label{sec:system}
We consider the following discrete-time complex LQO system:
\begin{empheq}[left={\mathcal{S}}:\empheqlbrace]{equation}
  \begin{aligned}
    x_{k} &= A x_{k-1} + B u_k,\\
    y_k &= C x_k + M\,\tred{\bigl(x_k \otimes \overline{x_k}\bigr)},
  \end{aligned}
  \label{eq:d-LQO}
\end{empheq}
where $u_{k}\in \mathbb{C}^m$, $x_{k}\in \mathbb{C}^n$, and $y_{k}\in \mathbb{C}^{p}$ denote the input, state, and output, respectively, and
$A\in \mathbb{C}^{n\times n}$, $B\in \mathbb{C}^{n\times m}$, $C\in \mathbb{C}^{p\times n}$,
$M=\big[\,\mathrm{vec}(M_1),\ldots,\mathrm{vec}(M_p)\,\big]^{\top}\in\mathbb{C}^{p\times n^2}$.
For any $i=1,\dots,p$, let $U_i\in \mathbb{C}^{c\times n}$ and $M_i=U_i^{\ast}U_i\in \mathbb{C}^{n\times n}$ be Hermitian.
\tred{Thus, the $i$-th quadratic output is $x_k^\ast M_i x_k$.}
Assume that $A$ is Schur stable, i.e., all eigenvalues lie strictly inside the open unit disk.

The Volterra kernels of \eqref{eq:d-LQO} are
\begin{align}
  h_1[t] &= C A^{t} B, \quad
  h_2[t_1,t_2] = M\tred{\bigl(A^{t_1}B\otimes\overline{A^{t_2}B}\bigr)},
\end{align}
and, with $x_{-1}=0$, \eqref{eq:d-LQO} is equivalent to the convolution representation
\begin{align}\label{eq:conv_dssm}
  y_k&=\sum_{t=0}^{k} h_1[t]\cdot u_{k-t}
  \;+\;\sum_{t_1=0}^{k}\sum_{t_2=0}^{k} h_2[t_1,t_2]\cdot\tred{\bigl(u_{k-t_1}\otimes \overline{u_{k-t_2}}\bigr)}.
\end{align}
On the finite horizon $\{0,1,\dots,L-1\}$, the time-limited $h^2$-norm $\|\mathcal{S}\|_{h^2_{L}}$ of \eqref{eq:d-LQO} is defined by
\begin{align}
    \|\mathcal{S}\|_{h^2_{L}}^{2}:=\|h_1\|_{\ell^2_{L}}^{2}+\|h_2\|_{\ell^2_{L}}^{2}.
\end{align}
Let $P_L$ be the solution of the finite-horizon Lyapunov equation
\begin{align}
    P_L=AP_LA^\ast+BB^\ast-A^{L}BB^{\ast}(A^{\ast})^{L}.
\end{align}
\tblack{By adapting arguments from \cite{reiter2025h,zulfiqar2024time} for continuous-time LQO systems to the discrete-time setting, one can verify that}
\begin{align}
    \|h_1\|_{\ell^2_{L}}^{2} = \mathrm{tr}\big(C P_L C^\ast\big), \:\: \|h_2\|_{\ell^2_{L}}^{2} = \sum_{k=1}^{p}\mathrm{tr}\big(P_L M_k P_L M_k\big).
    \label{eq:h2_lyapunov_formula}
\end{align}

A reduced-order model (ROM) for \eqref{eq:d-LQO} is given by
\begin{empheq}[left={\hat{\mathcal{S}}}:\empheqlbrace]{equation}
  \begin{aligned}
    \hat{x}_{k} &= \hat{A}\hat{x}_{k-1} + \hat{B}u_k,\\
    \hat{y}_k &= \hat{C}\hat{x}_k + \hat{M}\,\tred{\bigl(\hat{x}_k \otimes \overline{\hat{x}_k}\bigr)},
  \end{aligned}
  \label{eq:d-r-LQO}
\end{empheq}
where $u_{k}\in \mathbb{C}^m$, $\hat{x}_{k}\in \mathbb{C}^r$, $\hat{y}_{k}\in \mathbb{C}^{p}$,
$\hat{A}\in \mathbb{C}^{r\times r}$, $\hat{B}\in \mathbb{C}^{r\times m}$, $\hat{C}\in \mathbb{C}^{p\times r}$,
$\hat{M}=\big[\,\mathrm{vec}(\hat{M}_1),\ldots,\mathrm{vec}(\hat{M}_p)\,\big]^{\top}\in\mathbb{C}^{p\times r^2}$,
with $\hat{A}$ stable and each $\hat{M}_i\in \mathbb{C}^{r\times r}$ Hermitian.

\tblack{One can derive a discrete-time, time-limited analogue of the output error bound from~\cite{benner2021gramians} for the continuous-time case applied to \eqref{eq:d-LQO} and \eqref{eq:d-r-LQO} as follows:}
\begin{align}
    &\|y-\hat{y}\|_{\ell_{L}^{\infty}}^{2}\\
    &\leq \big(\|h_1-\hat{h}_{1}\|_{\ell_{L}^{2}}^2 + \|h_2-\hat{h}_{2}\|_{\ell_{L}^{2}}^2\big)\,(1 + \|u\|_{\ell_{L}^{2}}^2)\,\|u\|_{\ell_{L}^{2}}^2 \nonumber\\
    &= \|\mathcal{S}-\hat{\mathcal{S}}\|_{h_{L}^{2}}^2 \,(1 + \|u\|_{\ell_{L}^{2}}^2)\,\|u\|_{\ell_{L}^{2}}^2.
    \label{eq:single_h2_ineq}
\end{align}
Let $\tilde P_L$ and $\hat P_L$ be the solutions to the finite-horizon Sylvester/Lyapunov equations
\begin{align}
    \tilde P_L &= A\tilde P_L\hat A^\ast + B\hat B^\ast - A^{L}B\hat B^\ast(\hat A^\ast)^{L},\label{eq:Ptilde_t}\\
    \hat P_L &= \hat A\hat P_L\hat A^\ast + \hat B\hat B^\ast - \hat A^{L}\hat B\hat B^\ast(\hat A^{\ast})^{L}.\label{eq:Phat_t}
\end{align}
\tblack{Substituting the formulas in \eqref{eq:h2_lyapunov_formula} into
$\|\mathcal{S}-\hat{\mathcal{S}}\|_{h_L^2}^2=\|h_1-\hat h_1\|_{\ell_L^2}^2+\|h_2-\hat h_2\|_{\ell_L^2}^2$, we obtain}
\begin{align}
\|\mathcal{S}-\hat{\mathcal{S}}\|_{h_{L}^{2}}^{2}
&= \mathrm{tr}\big(C P_L C^{*}\big)
+\mathrm{tr}\big(\hat C \hat P_L \hat C^{*}\big)
-2\,\mathrm{Re}\,\mathrm{tr}\big(C \tilde P_L \hat C^{*}\big) \nonumber\\
&\quad
+\sum_{k=1}^{p}\Big(
\mathrm{tr}\big(P_L M_k P_L M_k\big)
+\mathrm{tr}\big(\hat P_L \hat M_k \hat P_L \hat M_k\big) \nonumber\\
&\qquad\qquad\quad-2\,\mathrm{Re}\,\mathrm{tr}\big(\tilde P_L^{*} M_k \tilde P_L \hat M_k\big)
\Big).
\label{eq:single_h2_error}
\end{align}
From \eqref{eq:single_h2_ineq}, for any sufficiently small input norm, making the $h_{L}^{2}$ error \eqref{eq:single_h2_error} small guarantees a sufficiently small output error.
\tblack{Directly optimizing the output error on the left-hand side of~\eqref{eq:single_h2_ineq} is difficult (see Remark~\ref{rem:exi_difficult}).
Therefore, we instead minimize the system-level upper-bound term $\|\mathcal{S}-\hat{\mathcal{S}}\|_{h_L^2}^2$.}


\section{Problem setting}\label{sec:problem_setting}
\subsection{Deep SSMs with Linear Quadratic-Output Systems}
In this work, \tblack{with a view to performing computationally efficient MOR,} we consider a $\xi$-layer LQO-type Deep SSM as shown in Fig.~\ref{fig:structure_of_dssm}.
At the input layer, an input sequence $(s_k)_{k=0}^{L-1}$ with feature dimension $H$ is given and transformed into the $m$-dimensional input of the first intermediate layer, $s_{\mathrm{in}, k}=u_k^{(1)}\in\mathbb{R}^m$.
For a general $i$-th intermediate layer ($i=1,2,\dots,\xi$), the input $u_k^{(i)}$ is mapped to the output $y_k^{(i)}$ by the following complex LQO system corresponding to~\eqref{eq:d-LQO}:
\begin{empheq}[left={\mathcal{S}^{(i)}:\empheqlbrace}]{equation}
  \begin{aligned}
    x_{k}^{(i)} &= A^{(i)}x_{k-1}^{(i)} + B^{(i)}u_k^{(i)},\\
    y_k^{(i)} &= C^{(i)}x_k^{(i)} + M^{(i)}\tred{\bigl(x_k^{(i)}\otimes \overline{x_k^{(i)}}\bigr)},
  \end{aligned}
  \label{eq:d-LQO-dssm}
\end{empheq}
where $u_{k}^{(i)}\in \mathbb{R}^m$, $x_{k}^{(i)}\in \mathbb{C}^n$, $y_{k}^{(i)}\in \mathbb{C}^{m}$,
$A^{(i)}\in \mathbb{C}^{n\times n}$ is diagonal and Schur stable,
$B^{(i)}\in \mathbb{C}^{n\times m}$, $C^{(i)}\in \mathbb{C}^{m\times n}$,
$M^{(i)}\in \mathbb{C}^{m\times n^2}$, and for any $j=1,\dots,m$,
$U_j^{(i)}\in \mathbb{C}^{c\times n}$ with $M_j^{(i)}=(U_j^{(i)})^{\ast}U_j^{(i)}\in \mathbb{C}^{n\times n}$ Hermitian.
\tblack{Here, the diagonal structure of $A^{(i)}$ is not an additional diagonalizability assumption on an arbitrary dense matrix, but an architectural property of the target Deep SSM class.
In particular, $A^{(i)}=\operatorname{diag}(\Lambda^{(i)})$ is directly parameterized, as in existing Deep SSM architectures~\cite{gu2022efficiently,gu2022parameterization, smith2023simplified, gu2023mamba}, and no dense matrix is diagonalized during training, inference, or compression.}
For stable training, we apply a residual connection to the real parts of the input $u_k^{(i)}$ and the output $y_k^{(i)}$, followed by layer normalization (LN):
\begin{align}
  z_k^{(i)}&=u_k^{(i)}+\mathrm{Re}\bigl(y_k^{(i)}\bigr), \label{eq:dssm_residual}\\
  u_{k}^{(i+1)}&=\mathrm{LN}_{\gamma_1^{(i)},\,\gamma_2^{(i)}}\!\bigl(z_k^{(i)}\bigr), \label{eq:dssm_LN}
\end{align}
where, for $\varepsilon>0$, LN is defined by
\begin{align}
  &\mathrm{LN}_{\gamma_1^{(i)},\,\gamma_2^{(i)}}(z_k^{(i)})
  =\gamma_1^{(i)}\odot \frac{z_k^{(i)}-\mu_{k,i}}{\sigma_{k,i}}
  +\gamma_2^{(i)}, \label{eq:LN}\\
  &\mu_{k,i}:=\frac{1}{m}\mathbf{1}^{\top}z_{k}^{(i)},\quad
  \sigma_{k,i}:=\sqrt{\frac{1}{m}\bigl\|z_k^{(i)}-\mu_{k,i}\mathbf{1}\bigr\|^2+\varepsilon}.
\end{align}
Note that $\gamma_1^{(i)}$ and $\gamma_2^{(i)}$ are the learnable affine parameters of LN at layer $i$, representing the per-feature scale and bias, respectively.
The sequence $(s_{\mathrm{out}, k})_{k=0}^{L-1}:=(u_k^{(\xi+1)})_{k=0}^{L-1}$ is the final output of the intermediate layers and is then mapped by an output layer to the desired shape.

\begin{figure}[htbp]
 \centering
 \includegraphics[width=\linewidth]{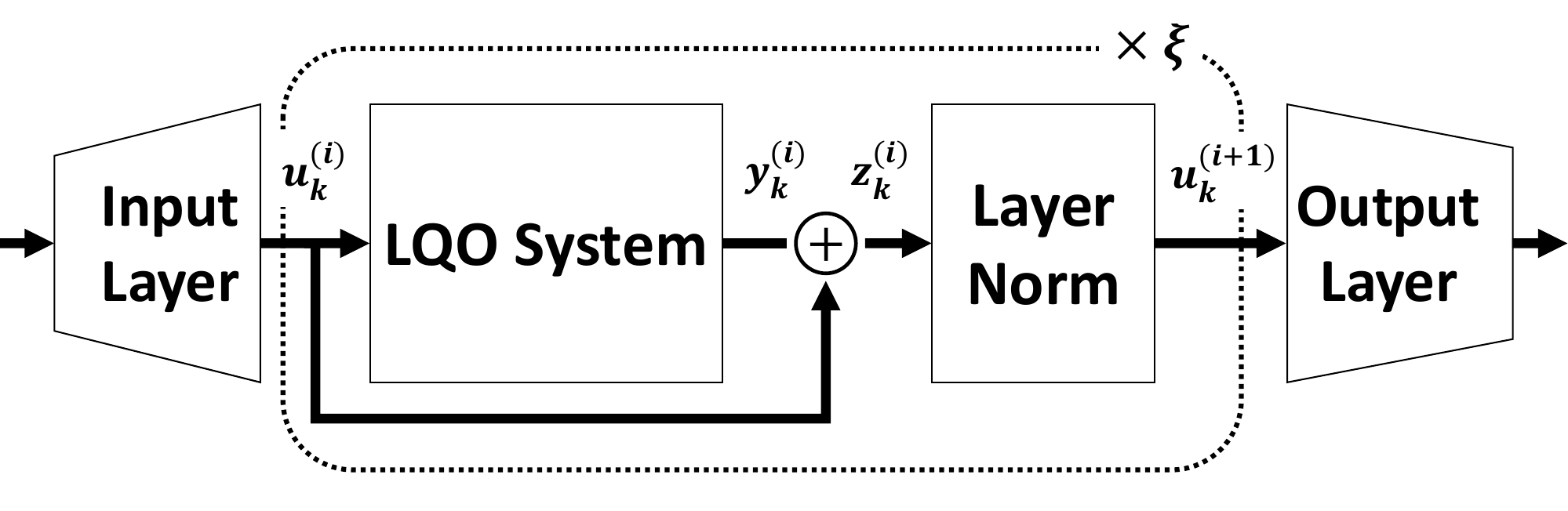}
 \caption{Structure of the Deep SSM considered in this work. \label{fig:structure_of_dssm}}
\end{figure}

\begin{remark}
The Deep SSM used in this study can be interpreted as an LQO-type extension of S5~\cite{smith2023simplified} when the activation immediately after each linear state-space block is the quadratic function $\sigma(x)=(|x_1|^2,\dots,|x_m|^2)\in\mathbb{R}^{m}$.
In S5, the input-output relation from $u_k$ to $y_k$ is expressed as
\begin{empheq}[left=\empheqlbrace]{equation}
  \begin{aligned}
    x_{k} &= A_{\mathrm{S5}}\,x_{k-1} + B_{\mathrm{S5}}\,u_k,\\
    y_k &= \sigma\bigl(C_{\mathrm{S5}}\,x_k\bigr).
  \end{aligned}
\end{empheq}
This relation can be written in the LQO form~\eqref{eq:d-LQO-dssm} by taking $A=A_{\mathrm{S5}}$, $B=B_{\mathrm{S5}}$, $C=0$, and
$M_j = C_{\mathrm{S5},\,j}^{\tred{\ast}}C_{\mathrm{S5},\,j}$, where $C_{\mathrm{S5},\,j}$ denotes the $j$-th row of $C_{\mathrm{S5}}$.
In this work, we model the input-output relation from $u_k$ to $y_k$ in Deep SSMs using more general $C$ and $M$ matrices.
\end{remark}

\begin{remark}
\tblack{The LQO system~\eqref{eq:d-LQO} can be viewed as a Wiener-type block-oriented system, where a linear dynamical system is followed by a static quadratic output map~\cite{SchoukensTiels2017Survey}.
Since the Deep SSM considered in this paper consists of cascaded LQO-type blocks, its training is related, in a broad sense, to block-oriented nonlinear system identification~\cite{SchoukensTiels2017Survey}.}
\end{remark}

\subsection{MOR-Based Compression of Deep SSMs}\label{subsec:mor_based_compression}
\tblack{Existing MOR-based compression methods for Deep SSMs~\cite{ezoe2024model,forgione2024model,gwak2024layer,sakamoto2025compression,chahine2026curious} reduce the internal SSMs layerwise using criteria such as BT, $\mathcal{H}^2$-based MOR, $\mathcal{H}^{\infty}$-based pruning scores, or in-training MOR.
Despite their different criteria, these methods aim to preserve the layerwise input-output behavior. However, how the layerwise $h^2$ approximation errors affect the overall output error and which reduction strategy is most effective from this viewpoint remain unclear.}

In this work, given an input sequence $(s_{\mathrm{in},k})_{k=0}^{L-1}$, we construct a reduced Deep SSM that approximates the output $(s_{\mathrm{out},k})_{k=0}^{L-1}$ of a $\xi$-layer Deep SSM.
We refer to the pretrained model whose layers are $n$-dimensional complex LQO systems~\eqref{eq:d-LQO} as an $n$-dimensional Deep SSM, and to the model whose $i$-th layer is replaced by the ROM~\eqref{eq:d-r-LQO} with $r=r_i$ as an $r_i$-dimensional Deep SSM.
Unlike prior MOR-based approaches for Deep SSMs~\cite{ezoe2024model, forgione2024model, gwak2024layer, sakamoto2025compression,chahine2026curious}, our goal is to construct an $r_i$-dimensional Deep SSM that minimizes an upper bound on the output error $\|s_{\mathrm{out}}-\hat{s}_{\mathrm{out}}\|_{\ell_{L}^{\infty}}$ for any input sequence $s_{\mathrm{in}}$; see Fig.~\ref{fig:comparison}.
\tblack{Section~\ref{sec:experiment} shows that, unlike the methods in~\cite{ezoe2024model, sakamoto2025compression}, this enables the construction of a high-performance compressed model even without retraining.} 
Directly optimizing this output error is difficult; see Remark~\ref{rem:exi_difficult}.

\section{Output Error Analysis for Deep SSMs}\label{sec:error_analysis}
For any input sequence $(s_{\mathrm{in},k})_{k=0}^{L-1}$, we aim to construct a reduced Deep SSM that minimizes the output error
\begin{align}\label{eq:exi}
    e_{\xi}:=\|s_{\mathrm{out}}-\hat{s}_{\mathrm{out}}\|_{\ell_{L}^{\infty}}
\end{align}
of Deep SSMs, and derive an upper bound on $e_{\xi}$ that depends on the input.
Hereafter, for the $i$-th intermediate layer of the $n$-dimensional Deep SSM and the $r_i$-dimensional Deep SSM, denote their LQO inputs/outputs by $(u_k^{(i)}, y_k^{(i)})$ and $(\hat{u}_k^{(i)}, \hat{y}_k^{(i)})$, respectively, for all $k$, and assume $u_k^{(1)}=\hat{u}_k^{(1)}(=s_{\mathrm{in},k})$.


\begin{theorem}\label{thm:error_analysis_general}
Let $\mathcal{S}^{(i)}=(A^{(i)}, B^{(i)}, C^{(i)}, M^{(i)})$ denote the $n$-dimensional Deep SSM at the $i$-th layer of the Deep SSM in~\eqref{eq:d-LQO-dssm}, and let $\hat{\mathcal{S}}^{(i)}=(\hat{A}^{(i)}, \hat{B}^{(i)}, \hat{C}^{(i)}, \hat{M}^{(i)})$ denote the $r=r_i$-dimensional SSM at the $i$-th layer of the $r_i$-dimensional Deep SSM.
Then
\begin{align}
  e_{\xi}\;&\leq\;\sum_{i=1}^{\xi} G_i\,
  \|\mathcal{S}^{(i)}-\hat{\mathcal{S}}^{(i)}\|_{h^2_L}\,\cdot\bigl(\,\|\hat{u}^{(i)}\|_{\ell_{L}^2}\,\sqrt{1+\|\hat{u}^{(i)}\|_{\ell_{L}^{2}}^{2}}\,\bigr),\label{eq:dssm_output_error}
\end{align}
where, with $\omega := \max_{i}\bigl(\operatorname{Lip}(\mathrm{LN}_{\gamma_1^{(i)},\gamma_2^{(i)}})\bigr)$ denoting the maximum Lipschitz constant of LN~\eqref{eq:LN}, 
\begin{align}
  G_i &= \omega^{\xi-i+1}\,\Bigl(\prod_{j=i+1}^{\xi} g_j\Bigr),\\
  g_j &= 1 + \sqrt{L}\cdot\left\{\|h_1^{(j)}\|_{\ell_{L}^{2}} + \|h_2^{(j)}\|_{\ell_{L}^{2}}\bigl(\|u^{(j)}\|_{\ell_{L}^2} + \|\hat{u}^{(j)}\|_{\ell_{L}^2}\bigr)\right\}.
\end{align}
\end{theorem}

\begin{proof}
First, for the two LQO systems $\mathcal{S}^{(i)}$ and $\hat{\mathcal{S}}^{(i)}$ at layer $i$, let $e_{\mathrm{SSM},i}:=\|y^{(i)}-\hat{y}^{(i)}\|_{\ell_{L}^{\infty}}$ denote the output error.
From~\eqref{eq:conv_dssm},
\begin{align}
  &y_k^{(i)} - \hat y_k^{(i)} \\
  &=\sum_{\ell=0}^{k} h_1^{(i)}[\ell]\cdot u_{k-\ell}^{(i)} - \hat h_1^{(i)}[\ell]\cdot \hat u_{k-\ell}^{(i)} \nonumber\\
  &\quad + \sum_{\ell_1=0}^{k}\sum_{\ell_2=0}^{k} h_2^{(i)}[\ell_1,\ell_2]\cdot\tred{\bigl(u_{k-\ell_1}^{(i)}\otimes \overline{u_{k-\ell_2}^{(i)}}\bigr)} \nonumber\\
  &\qquad\qquad - \hat h_2^{(i)}[\ell_1,\ell_2]\cdot\tred{\bigl(\hat u_{k-\ell_1}^{(i)}\otimes \overline{\hat u_{k-\ell_2}^{(i)}}\bigr)} \nonumber\\
  &=\sum_{\ell=0}^{k} h_1^{(i)}[\ell]\cdot\bigl(u_{k-\ell}^{(i)} - \hat u_{k-\ell}^{(i)}\bigr)
   + \bigl(h_1^{(i)}[\ell] - \hat h_1^{(i)}[\ell]\bigr)\cdot \hat u_{k-\ell}^{(i)} \nonumber\\
  &\quad + \sum_{\ell_1=0}^{k}\sum_{\ell_2=0}^{k} h_2^{(i)}[\ell_1,\ell_2]\Bigl(\tred{u_{k-\ell_1}^{(i)}\otimes \overline{u_{k-\ell_2}^{(i)}}} - \tred{\hat u_{k-\ell_1}^{(i)}\otimes \overline{\hat u_{k-\ell_2}^{(i)}}}\Bigr) \nonumber\\
  &\quad + \bigl(h_2^{(i)}[\ell_1,\ell_2]-\hat h_2^{(i)}[\ell_1,\ell_2]\bigr)\cdot\tred{\bigl(\hat u_{k-\ell_1}^{(i)}\otimes \overline{\hat u_{k-\ell_2}^{(i)}}\bigr)}.
\end{align}
Hence, by the triangle inequality,
\begin{align}
  e_{\mathrm{SSM},i}
  &\le \|h_1^{(i)}\conv(u^{(i)}-\hat u^{(i)})\|_{\ell_{L}^{\infty}}
     + \|(h_1^{(i)}-\hat h_1^{(i)})\conv \hat u^{(i)}\|_{\ell_{L}^{\infty}} \nonumber\\
  &\quad+ \|h_2^{(i)}\conv \bigl(\tred{u^{(i)}\otimes \overline{u^{(i)}}} - \tred{\hat u^{(i)}\otimes \overline{\hat u^{(i)}}}\bigr)\|_{\ell_{L}^{\infty}} \nonumber\\
  &\quad+ \|(h_2^{(i)}-\hat h_2^{(i)})\conv \tred{(\hat u^{(i)}\otimes \overline{\hat u^{(i)}})}\|_{\ell_{L}^{\infty}}.
\end{align}
Let $\alpha_{1,i}:=\|h_1^{(i)}\|_{\ell_{L}^{2}}$, $\alpha_{2,i}:=\|h_2^{(i)}\|_{\ell_{L}^{2}}$, $\tilde \alpha_{1,i}:=\|h_1^{(i)}-\hat h_1^{(i)}\|_{\ell_{L}^{2}}$, $\tilde \alpha_{2,i}:=\|h_2^{(i)}-\hat h_2^{(i)}\|_{\ell_{L}^{2}}$,
$\beta_{i}:=\|u^{(i)}\|_{\ell_{L}^{2}}$, and $\hat\beta_{i}:=\|\hat u^{(i)}\|_{\ell_{L}^{2}}$.
By Young's convolution inequality and the relation $\|\cdot\|_{\ell_{L}^{2}}\leq \sqrt{L}\|\cdot\|_{\ell_{L}^{\infty}}$,
\begin{align}
    &\|h_1^{(i)}\conv(u^{(i)} - \hat u^{(i)})\|_{\ell_{L}^{\infty}}\leq \sqrt{L}\cdot \alpha_{1,i}\cdot e_{i-1} \\
    &\|(h_1^{(i)} - \hat h_1^{(i)})\conv \hat u^{(i)}\|_{\ell_{L}^{\infty}}\leq \tilde\alpha_{1,i}\cdot \hat\beta_{i}.
\end{align}
Also, we get
\begin{align}
    &\|h_2^{(i)}\conv (\tred{u^{(i)} \otimes \overline{u^{(i)}}} - \tred{\hat u^{(i)} \otimes\overline{\hat u^{(i)}}})\|_{\ell_{L}^{\infty}} \\
    &\leq \alpha_{2,i}\cdot \|\tred{u^{(i)} \otimes \overline{u^{(i)}}} - \tred{\hat u^{(i)} \otimes\overline{\hat u^{(i)}}}\|_{\ell_{L}^{2}} \\
    &\leq \sqrt{L}\cdot \alpha_{2,i}\cdot e_{i-1}\cdot (\beta_i + \hat\beta_i),
\end{align} 
and 
\begin{align}
    \|(h_2^{(i)} - \hat h_2^{(i)})\conv \tred{(\hat u^{(i)} \otimes\overline{\hat u^{(i)}})}\|_{\ell_{L}^{\infty}}
    &\leq \tilde \alpha_{2,i}\cdot \|\tred{\hat u^{(i)} \otimes \overline{\hat u^{(i)}}}\|_{\ell_{L}^{2}} \\
    &= \tilde \alpha_{2,i}\cdot \hat\beta_i^{2}.
\end{align}
\tred{Here, $\|\overline{u}\|_{\ell_L^2}=\|u\|_{\ell_L^2}$, so the complex conjugation does not change the resulting upper bound.}
Thus, with $\kappa_i:=\sqrt{L}\cdot (\alpha_{1,i}+\alpha_{2,i}(\beta_i+\hat\beta_i))$,
\begin{align}
  e_{\mathrm{SSM},i}
  &\le \kappa_i\,e_{i-1} + \tilde\alpha_{1,i}\,\hat\beta_i + \tilde\alpha_{2,i}\,\hat\beta_i^{2} \nonumber\\
  &\le \kappa_i\,e_{i-1}
   + \sqrt{\tilde \alpha_{1,i}^{2}+\tilde \alpha_{2,i}^{2}}\;\sqrt{1+\hat\beta_i^{2}}\;\hat\beta_i.
\end{align}
Taking the residual connection~\eqref{eq:dssm_residual} and LN~\eqref{eq:dssm_LN} into account, we obtain the recurrence for the output error $e_{i}$:
\begin{align}
  e_{i}
  &\le \operatorname{Lip}\bigl(\mathrm{LN}_{\gamma_1^{(i)},\gamma_2^{(i)}}\bigr)\,\|z^{(i)}-\hat z^{(i)}\|_{\ell_{L}^{\infty}} \nonumber\\
  &\le \operatorname{Lip}\bigl(\mathrm{LN}_{\gamma_1^{(i)},\gamma_2^{(i)}}\bigr)\,(e_{\mathrm{SSM},i}+e_{i-1}) \nonumber\\
  &\le \operatorname{Lip}\bigl(\mathrm{LN}_{\gamma_1^{(i)},\gamma_2^{(i)}}\bigr)\nonumber\\
  &\quad\cdot\Bigl\{(1+\kappa_i)\,e_{i-1}
  + \sqrt{\tilde \alpha_{1,i}^{2}+\tilde \alpha_{2,i}^{2}}\cdot\sqrt{1+\hat\beta_i^{2}}\cdot\hat\beta_i\Bigr\}. \;\;\;\label{eq:reccurence}
\end{align}
Using the relation $\|\mathcal{S}^{(i)}-\hat{\mathcal{S}}^{(i)}\|_{h^2_L}=(\tilde \alpha_{1,i}^{2}+\tilde \alpha_{2,i}^{2})^{1/2}$ and solving the recurrence~\eqref{eq:reccurence} with $e_0=0$, we obtain \eqref{eq:dssm_output_error}.
\end{proof}

Because~\eqref{eq:LN} normalizes the state-space-block outputs, the input norms $\|u^{(j)}\|_{\ell_{L}^2}$ and $\|\hat{u}^{(j)}\|_{\ell_{L}^2}$ can be regarded as bounded, yielding the following corollary.
\begin{corollary}\label{cor:constant_bound}
Let $b$ be a constant and assume $\max_{1\le j\le \xi}\bigl(\|u^{(j)}\|_{\ell_{L}^2},\,\|\hat{u}^{(j)}\|_{\ell_{L}^2}\bigr)\le b$.
Under the assumptions of Theorem~\ref{thm:error_analysis_general},
\begin{align}\label{eq:cascade_h2_error}
  e_{\xi}
  \;\le\; \bigl(b\sqrt{1+b^2}\bigr)\;\sum_{i=1}^{\xi}\tilde{G}_i\,
  \|\mathcal{S}^{(i)}-\hat{\mathcal{S}}^{(i)}\|_{h^2_L},
\end{align}
where 
\begin{align}
  \tilde{G}_i &= \omega^{\xi-i+1}\,\Bigl(\prod_{j=i+1}^{\xi}\tilde{g}_j\Bigr),\\
  \tilde{g}_j &= 1 + \sqrt{L}\left(\|h_1^{(j)}\|_{\ell_{L}^{2}} + 2b\,\|h_2^{(j)}\|_{\ell_{L}^{2}}\right).
\end{align}
\end{corollary}
\begin{proof}
This follows from Theorem~\ref{thm:error_analysis_general}.
\end{proof}
\tblack{Inequality~\eqref{eq:cascade_h2_error} shows that, for the same sufficiently small input $s_{\mathrm{in}}$, the final output error $e_{\xi}$ is bounded by a linear combination of the $h^2$ errors between the two LQO systems placed at each layer.
\tred{Here, the bounds in Theorem~\ref{thm:error_analysis_general} and Corollary~\ref{cor:constant_bound} can be conservative, especially because they involve products of LayerNorm Lipschitz constants and layerwise kernel norms, and are not claimed to be tight.}
On the other hand, as shown in Section~\ref{sec:experiment} and Appendix~\ref{app:wikitext103}, by constructing ROMs so that the upper bound becomes small based on the insight obtained from the following lemma, \tred{one can construct a reduced Deep SSM with a small output error relative to the pretrained model.}}
An additional comparison of normalized empirical errors and the bound is given in Appendix~\ref{app:normalized_bound}.

\begin{lemma}\label{lem:LN-Lip}
Let $m\ge 2$ and $\varepsilon>0$.
Denote by $\operatorname{Lip}(\mathrm{LN}_{\gamma_1,\gamma_2})$ 
the Lipschitz constant of $\mathrm{LN}$ in \eqref{eq:LN}. Then, for $\gamma_1,\gamma_2\in\mathbb{R}^m$,
\begin{align}
\frac{\|\gamma_1\|_{{\infty}}}{\sqrt{\varepsilon}}\sqrt{1-\frac{1}{m}}
\;\;\le\;\;
\operatorname{Lip}(\mathrm{LN}_{\gamma_1,\gamma_2})
\;\;\le\;\;
\frac{\|\gamma_1\|_{{\infty}}}{\sqrt{\varepsilon}}\,.
\end{align}
\end{lemma}
\begin{proof}
See Appendix~\ref{app:proofs}.
\end{proof}

Lemma~\ref{lem:LN-Lip} characterizes the contribution coefficient $\tilde G_i$.
Taking $\varepsilon$ in~\eqref{eq:LN} sufficiently small gives $\operatorname{Lip}(\mathrm{LN}_{\gamma_1,\gamma_2})\gg1$.
Moreover, $\tilde g_j\ge1$ for every $j$.
Hence, Corollary~\ref{cor:constant_bound} gives $\tilde G_i\ge\tilde G_{i+1}$ for every $i$.
This means that $h^2_{L}$-norm errors in shallower layers have a stronger effect on \tred{the derived upper bound for} the final intermediate output error of the Deep SSM.


\section{Model reduction for Deep SSM}\label{sec:proposed_mor}
\subsection{Model Order Reduction Problem}
Based on Corollary~\ref{cor:constant_bound}, we formulate the following MOR problem, which minimizes the upper bound in~\eqref{eq:cascade_h2_error}.

Let $r_{\mathrm{list}}:=[\,r_1,r_2,\ldots,r_{\xi}\,]\in \mathbb{N}^{\xi}$.
For each $i=1,\ldots,\xi$, let
$\hat U^{(i)} := [\,\mathrm{vec}(\hat U_1^{(i)}),\ldots,\mathrm{vec}(\hat U_m^{(i)})\,]^{\top} \in\mathbb{C}^{m\times c r_i}$, and define
$\mathcal{C}_i:=\{ \hat\Lambda^{(i)}\in\mathbb{C}^{r_i}\mid |\hat\Lambda^{(i)}_j|<1,\ j=1,\dots,r_i\}\times \mathbb{C}^{r_i\times m} \times \mathbb{C}^{m\times r_i} \times \mathbb{C}^{m\times c r_i}$.
Let $\hat{\mathbf{S}} := (\hat{\mathcal{S}}^{(i)})_{i=1}^{\xi}$.
The corresponding parameter representation is given by $\hat\theta:=\{(\hat{\Lambda}^{(i)},\hat B^{(i)},\hat C^{(i)},\hat U^{(i)})\}_{i=1}^{\xi}$.

Then, for a given $r_{\mathrm{list}}$, the MOR problem can be written as
\begin{equation}\label{eq:opt_original}
\begin{aligned}
  \text{minimize}\quad
  & f(\hat{\mathbf{S}}):=\sum_{i=1}^{\xi}\,\tilde{G}_i\,\bigl\|\,\mathcal{S}^{(i)}-\hat{\mathcal{S}}^{(i)}\,\bigr\|_{h^2_L} \\
  \text{subject to}\quad
  & (\hat{\Lambda}^{(i)}, \hat{B}^{(i)}, \hat{C}^{(i)}, \hat U^{(i)}) \in \mathcal{C}_i,\quad \forall i=1,\dots,\xi.
\end{aligned}
\end{equation}
Here, $\mathcal{S}^{(i)}$ and $\hat{\mathcal{S}}^{(i)}$ denote, respectively for the $i$-th layer of the Deep SSM, the systems~\eqref{eq:d-LQO} with $A:=\mathrm{diag}(\Lambda)$ and~\eqref{eq:d-r-LQO} with $\hat{A}:=\mathrm{diag}(\hat{\Lambda})$.

\tblack{
The design of $r_{\mathrm{list}}$ is important when solving the MOR problem.
A larger $r$ generally improves ROM accuracy by retaining more dynamical information.
Therefore, assigning larger reduced orders to shallower layers is a reasonable guideline for reducing~\eqref{eq:opt_original}.
\tred{In practice, one may first fix the total reduced order according to a desired compression ratio or computational budget, use the bound to guide its allocation among layers, and then select the final allocation using validation performance.}
We show in Section~\ref{sec:experiment} that this design principle is also effective for model compression.}

Because $\tilde G_i$ is independent of the optimization variables, optimizing~\eqref{eq:opt_original} reduces to minimizing each $\|\mathcal S^{(i)}-\hat{\mathcal S}^{(i)}\|_{h_L^2}$, as in Fig.~\ref{fig:comparison}.
Note, as discussed in~\cite{sakamoto2025compression}, that because the internal linear SSMs in the Deep SSM are (i) defined on a finite horizon, (ii) required to be stability-guaranteed for reliable training, and (iii) formulated over the complex field, recent $\mathcal{H}^2$ MOR theory for LQO systems~\cite{reiter2025h, zulfiqar2024time} is not directly applicable in some cases.

\begin{remark}\label{rem:exi_difficult}
  In Deep SSM model reduction, optimizing $f$ in~\eqref{eq:opt_original} is preferable to directly optimizing $e_\xi$ in~\eqref{eq:exi}.
  The output error $e_\xi$ is nonsmooth, lacks a closed-form characterization, and involves multi-fold products of the layerwise Volterra kernels $h_1,h_2$.
  In contrast, $f$ has a closed form and is differentiable as a real-valued function of the real and imaginary parts except where the corresponding $h_L^2$ error is zero.
  It also enables ROMs that preserve the Deep SSM-specific properties in~\eqref{eq:opt_original}.
  Thus, rather than minimizing $e_\xi$ directly, we solve~\eqref{eq:opt_original} for reduced parameters that decrease it.
\end{remark}

\begin{remark}
  \tblack{The design principle, namely assigning larger reduced dimensions to shallow layers, is also useful for MOR-based compression for Deep SSM architectures~\cite{gu2022efficiently,gu2022parameterization, smith2023simplified, gu2023mamba} with general nonlinear functions such as GeLU. This is because, when the nonlinear functions are Lipschitz continuous, an output error bound can be derived in the same manner as Theorem~\ref{thm:error_analysis_general}. 
  However, for such general models, the corresponding model reduction problem may be computationally expensive.}
\end{remark}

\subsection{\texorpdfstring{Gradients for \eqref{eq:opt_original}}{Gradients}}\label{subsec:gradients}
Following~\cite{zulfiqar2024time}, we derive the gradients for~\eqref{eq:opt_original}. \tred{Unlike the continuous-time case, the discrete-time formulation involves Stein equations rather than continuous-time Lyapunov equations. The gradient formulas are derived using the corresponding discrete-time matrix equations.}

\begin{lemma}[\!\!\!\cite{zulfiqar2024time}]\label{lem:discrete_lqo_gradients}
Let the time-limited $h^2$ error between the complex LQO system~\eqref{eq:d-LQO} and its reduced system~\eqref{eq:d-r-LQO} be
$\varphi(\hat{\mathcal S})\;:=\;\|\mathcal S-\hat{\mathcal S}\|^2_{h^2_L}$.
Then the gradients of $\varphi(\hat{\mathcal S})$ are given by
\begin{align}
\nabla_{\hat A}\varphi
&=2\bigl(-(\tilde Y_L+2\tilde Z_L)^\ast A\tilde P_L+(\hat Y_L+2\hat Z_L)\hat A \hat P_L+L_L\bigr)\ , \\
\ \nabla_{\hat B}\varphi
&=2\bigl(-(\tilde Y_L+2\tilde Z_L)^\ast B+(\hat Y_L+2\hat Z_L)\hat B\bigr)\, , \\
\ \nabla_{\hat C}\varphi
&=2\bigl(-C\,\tilde P_L+\hat C\,\hat P_L\bigr)\, , \\
\nabla_{\hat M_i}\varphi
&=2\bigl(-\tilde P_L^\ast M_i\,\tilde P_L+\hat P_L\,\hat M_i\,\hat P_L\bigr),\quad i=1,\dots,p\ ,
\end{align}
where $\tilde P_L$ and $\hat P_L$ solve~\eqref{eq:Ptilde_t} and~\eqref{eq:Phat_t}, respectively.
Let $S_L:=A^L$ and $\hat S_L:=\hat A^L$.
Then $\tilde Y_L,\hat Y_L,\tilde Z_L,\hat Z_L$ solve the finite-horizon Lyapunov/Sylvester equations
\begin{align}
\tilde{Y}_L &= A^\ast \tilde{Y}_L \hat{A}+C^\ast \hat{C}-S^\ast_{L}C^\ast \hat{C}\hat{S}_{L},\label{eq:Ytil_t}\\
\hat Y_L &= \hat A^\ast\hat Y_L\hat A+\hat C^\ast\hat C-\hat{S}^{\ast}_{L}\hat C^\ast\hat C\hat{S}_{L},\label{eq:Yhat_t}\\
\tilde Z_L &= A^\ast \tilde Z_L\hat A+\sum_{i=1}^p\bigl(M_i\tilde P_L\hat M_i-S^\ast_{L}M_i\tilde P_L\hat M_i\hat S_{L}\bigr), \label{eq:Ztil_t}\\
\hat Z_L &= \hat A^\ast\hat Z_L\hat A+\sum_{i=1}^p\bigl(\hat M_i\hat P_L\hat M_i-\hat S^\ast_{L}\hat M_i\hat P_L\hat M_i\hat S_{L}\bigr).\label{eq:Zhat_t}
\end{align}
Moreover, with $\mathcal T^{\ast}_{\hat A,L}(X):=\sum_{j=0}^{L-1}(\hat A^\ast)^{j}\,X\,(\hat A^\ast)^{L-1-j}$ and
\begin{align}
    \tilde P &= A\tilde P\hat A^\ast+B\hat B^\ast, \label{eq:Ptil_inf}\\
    \hat P &= \hat A\hat P\hat A^\ast+\hat B\hat B^\ast ,\label{eq:Phat_inf}
\end{align}
we define
\begin{align}
&\tred{\tilde W_L:=C^\ast\hat C+2\sum_{i=1}^{p}M_i\tilde P_L\hat M_i},\\
&\tred{\hat W_L:=\hat C^\ast\hat C+2\sum_{i=1}^{p}\hat M_i\hat P_L\hat M_i},\\
&\tred{L_L:=-(\tilde Y_L+2\tilde Z_L)^\ast A(\tilde P-\tilde P_L)}\nonumber\\
&\tred{\qquad\quad +(\hat Y_L+2\hat Z_L)\hat A(\hat P-\hat P_L)}\nonumber\\
&\tred{\qquad\quad +\mathcal T^{\ast}_{\hat A,L}\!\left(\tilde W_L^\ast S_L\tilde P-\hat W_L\hat S_L\hat P\right).}\label{eq:Z_t}
\end{align}
\end{lemma}

\tred{The expression for $L_L$ follows by differentiating $\tilde P_L=\tilde P-S_L\tilde P\hat S_L^\ast$ and $\hat P_L=\hat P-\hat S_L\hat P\hat S_L^\ast$, and using the corresponding Stein equations and the Fr\'echet derivative of the matrix power. It was also verified against an independent finite-sum gradient and real- and imaginary-part finite differences for small complex matrices. A detailed derivation is given in Appendix~\ref{app:lemma2_derivation}.}

\begin{theorem}\label{thm:gradients}
Let $\varphi^{(i)}(\hat{\mathcal S}^{(i)})\;:=\;\|\mathcal S^{(i)}-\hat{\mathcal S}^{(i)}\|^2_{h^2_L}$ and $K_i:=2\sqrt{\varphi^{(i)}(\hat{\mathcal S}^{(i)})}$.
\tred{For $\varphi^{(i)}(\hat{\mathcal S}^{(i)})>0$,} the gradients of the MOR objective~\eqref{eq:opt_original} are
\begin{align}
\nabla_{\hat \Lambda^{(i)}}f
&=\frac{\tilde{G}_i}{K_i}\cdot\mathrm{diag}\!\left(\nabla_{\hat A^{(i)}}\varphi^{(i)}\right), \\
\nabla_{\hat B^{(i)}}f
&=\frac{\tilde{G}_i}{K_i}\cdot\nabla_{\hat B^{(i)}}\varphi^{(i)}, \qquad \nabla_{\hat C^{(i)}}f
=\frac{\tilde{G}_i}{K_i}\cdot\nabla_{\hat C^{(i)}}\varphi^{(i)}, \\
\nabla_{\hat U_j^{(i)}} f
&=
\hat U_j^{(i)}
\left(
\nabla_{\hat M_j^{(i)}}f+
(\nabla_{\hat M_j^{(i)}}f)^\ast
\right), \qquad j=1,\dots,m,
\end{align}
where $\nabla_{{\hat M_{j}^{(i)}}}f
=\frac{\tilde{G}_i}{K_i}\cdot\nabla_{{\hat M_{j}^{(i)}}}\varphi^{(i)}$.
\end{theorem}
\begin{proof}
\tblack{
Let $k$ be the total degrees of freedom. 
Under the identification $\mathbb{C}^k\simeq\mathbb{R}^{2k}$, we can use the standard real chain rule.}
By the chain rule,
$\mathrm{d}f=\sum_i \frac{\tilde G_i}{K_i}\,\mathrm{d}\varphi^{(i)}
=\sum_i\big\langle \tfrac{\tilde G_i}{K_i}\nabla_{\theta}\varphi^{(i)},\,\mathrm{d}\theta\big\rangle$ with $\theta\in\{\hat B^{(i)},\hat C^{(i)},\hat M^{(i)}_j\}$.
Moreover, since $\hat A^{(i)}=\mathrm{diag}(\hat\Lambda^{(i)})$, $\mathrm{d}\varphi^{(i)}=\langle\nabla_{\hat A^{(i)}}\varphi^{(i)},\mathrm{d}\hat A^{(i)}\rangle
=\langle\mathrm{diag}(\nabla_{\hat A^{(i)}}\varphi^{(i)}),\mathrm{d}\hat\Lambda^{(i)}\rangle$, which yields the stated expressions.
Finally, since $\hat M_j^{(i)}=(\hat U_j^{(i)})^\ast\hat U_j^{(i)}$, we have $\mathrm{d}\hat M_j^{(i)}=(\mathrm{d}\hat U_j^{(i)})^\ast\hat U_j^{(i)}+(\hat U_j^{(i)})^\ast\mathrm{d}\hat U_j^{(i)}$, and thus the real chain rule yields $\nabla_{\hat U_j^{(i)}} f$.
\end{proof}

\begin{remark}
\tblack{The gradient formulas in Theorem~\ref{thm:gradients} can be viewed as discrete-time, complex-valued counterparts of those in~\cite{zulfiqar2024time}, with the additional diagonal structure imposed on the reduced matrices.
Since the objective function $\psi$ is a real-valued function of complex variables, it is not complex analytic in general. Therefore, the gradients in Theorem~\ref{thm:gradients} should be understood as real Euclidean gradients with respect to the real and imaginary parts of the complex variables, written in complex matrix form. Equivalently, for a real-valued function $\psi$ of a complex variable $\widehat\theta$, Wirtinger calculus gives $\nabla_{\widehat\theta}\psi=2\frac{\partial \psi}{\partial \overline{\widehat\theta}}$, where $\overline{\widehat\theta}$ denotes the complex conjugate of $\widehat\theta$. Thus, the factor $2$ can be absorbed into the step size, and the resulting gradient-based algorithm is essentially unchanged.}
\end{remark}

\subsection{Gradient-Based Algorithm}\label{subsec:algorithm}
Alg.~\ref{alg:disc_cascade_tlh2_lqo_grad_short} is derived from Theorem~\ref{thm:gradients}.
It uses backtracking after each gradient computation so that the proposed reduced model satisfies the stability and Armijo conditions.
\tblack{Here, the Armijo condition refers to the sufficient decrease condition $f(\tilde\theta)\le f_{\ell}-c_{1}D_\ell$ in Alg.~\ref{alg:disc_cascade_tlh2_lqo_grad_short}.}
Since the $A$-matrices in~\eqref{eq:d-LQO} are diagonal, \tred{the large-scale Sylvester equations~\eqref{eq:Ytil_t}, \eqref{eq:Ztil_t}, and \eqref{eq:Ptil_inf} can be solved elementwise, allowing efficient gradient computation without a general Sylvester solver.}

\begin{theorem}\label{thm:convergence-ssm}
Let $\{\hat \theta_\ell\}_{\ell\ge 0}$ with $\hat \theta_\ell := \Bigl\{\,\bigl(\hat\Lambda^{(i)}_\ell,\; \hat B^{(i)}_\ell,\; \hat C^{(i)}_\ell,\; \hat U^{(i)}_{\ell}\bigr)\,\Bigr\}_{i=1}^{\xi}$ be the sequence generated by Alg.~\ref{alg:disc_cascade_tlh2_lqo_grad_short}.
Assume that $\{\hat \theta_\ell\}$ is bounded and, for every $i\in\{1,\dots,\xi\}$, the limit 
$\lim_{\ell\to\infty}\operatorname{diag}(\hat \Lambda_\ell^{(i)})$ exists and satisfies stability. 
\tred{Assume further that $\inf_{\ell,i}\varphi^{(i)}(\hat{\mathcal S}_{\ell}^{(i)})>0$.}
Then $\{\hat \theta_\ell\}$ converges to a stationary point of problem~\eqref{eq:opt_original}.
\end{theorem}
\begin{proof}
See Appendix~\ref{app:proofs}.
\end{proof}

Because~\eqref{eq:opt_original} is nonconvex, the initial point for Alg.~\ref{alg:disc_cascade_tlh2_lqo_grad_short} is important.
\tblack{Existing MOR methods for LQO or quadratic-output systems~\cite{PulchNarayan2019BT,GoseaAntoulas2019TwoSided,ReiterGoseaDuffGugercin2025H2} can provide initial reduced models; in particular, the time-limited methods~\cite{song2024balanced,zulfiqar2024time} are compatible with this setting.
On the other hand, these methods do not necessarily guarantee the stability of the reduced model, and in such cases, as discussed in Section~\ref{sec:experiment}, the performance of the compressed model can significantly deteriorate.
\tred{We also note that data-driven methods for LQO systems, such as AAA-LQO~\cite{GoseaGugercin2022AAA}, could in principle be used to construct initial reduced models from transfer-function samples generated using the available LQO realizations. However, their application would require a different modeling and experimental setup from the model-based compression considered in this study.}}
\tred{More generally, Alg.~\ref{alg:disc_cascade_tlh2_lqo_grad_short} can use a ROM from a future MOR method as its initializer if the ROM satisfies the feasibility constraints in~\eqref{eq:opt_original}, including Schur stability and the prescribed LQO parameterization.}

\begin{algorithm}[!htbp]
\caption{Gradient-based method for~\eqref{eq:opt_original}}
\label{alg:disc_cascade_tlh2_lqo_grad_short}
\begin{algorithmic}[1]
\REQUIRE
  Pretrained full-order models, initial ROM parameters $\hat\theta_0$, horizon length $L$, initial step sizes $\eta_{\mathrm{init}}=(\eta_{\Lambda},\eta_B,\eta_C,\eta_U)\tred{\in\mathbb{R}_{>0}^{4}}$, Armijo parameter~$c_1\tred{\in(0,1)}$, backtracking factor~$\rho\tred{\in(0,1)}$
\ENSURE
  $\hat\theta = \{(\hat\Lambda^{(i)}, \hat B^{(i)}, \hat C^{(i)}, (\hat U^{(i)}_j)_{j=1}^m)\}_{i=1}^{\xi}$
\FOR{$\ell=0,1,\ldots$}
  \STATE Set $f_\ell=f(\hat\theta_\ell)$ and compute $\nabla_{\hat\theta}f_\ell$ by Theorem~\ref{thm:gradients}.
  \STATE Set $\eta\leftarrow\eta_{\mathrm{init}}$.
  \WHILE{true}
    \STATE Propose $\tilde\theta=\hat\theta_\ell-\eta\odot\nabla_{\hat\theta}f_\ell$.
    \STATE \textbf{if} any $\diag(\tilde \Lambda^{(i)})$ is not Schur stable \textbf{then}
           $\eta_{\Lambda}\!\leftarrow\!\rho\,\eta_{\Lambda}$; \textbf{continue}
    \STATE Set $D_\ell:=\eta_{\Lambda}\sum_i\|\nabla_{\hat \Lambda^{(i)}}f_{\ell}\|^2
              +\eta_{B}\sum_i\|\nabla_{\hat B^{(i)}}f_{\ell}\|^2
              +\eta_{C}\sum_i\|\nabla_{\hat C^{(i)}}f_{\ell}\|^2
              +\eta_{U}\sum_{i,j}\|\nabla_{\hat U^{(i)}_j}f_{\ell}\|^2$.
    \STATE \textbf{if} $f(\tilde\theta)\le f_\ell-c_1D_\ell$ \textbf{then}
    $\hat\theta_{\ell+1}\leftarrow\tilde\theta$; \textbf{break}
    \STATE $\eta\leftarrow\rho\eta$.
  \ENDWHILE
\ENDFOR
\end{algorithmic}
\end{algorithm}


\section{Numerical experiments}\label{sec:experiment}
Using the bound in Section~\ref{sec:error_analysis} and the method in Section~\ref{sec:proposed_mor}, we compress models for the IMDb task of the Long Range Arena (LRA) dataset~\cite{tay2021long}.
\tblack{Specifically, we solve~\eqref{eq:opt_original} on the pretrained state-space blocks to decrease the upper bound on $e_\xi$ without retraining.}
The implementation is based on https://github.com/lindermanlab/S5.
\tblack{Appendix~\ref{app:wikitext103} reports additional WikiText-103 results.}

We pretrain the Deep SSM in Section~\ref{sec:problem_setting}.
We set $n=128$, $m=64$, $c=1$, $L=4096$, $\xi=4$, and the number of epochs to 30, while setting the other hyperparameters according to~\cite{smith2023simplified}.
During training and inference, we use linear SSMs discretized by zero-order hold.
The resulting model has 207,490 trainable parameters and achieves 86.66\% IMDb test accuracy.

We test several layerwise reduced dimensions $r_{\mathrm{list}}$, especially those allocating larger reduced dimensions to shallower layers, so as to reduce the upper bound in~\eqref{eq:opt_original}. 
For each $r_{\mathrm{list}}$, we compare Time-Limited Balanced Truncation (TLBT)~\cite{song2024balanced}, Time-Limited $\mathcal{H}^2$ model reduction (TLH2)~\cite{zulfiqar2024time}, and Alg.~\ref{alg:disc_cascade_tlh2_lqo_grad_short} initialized by TLBT or TLH2.
For Alg.~\ref{alg:disc_cascade_tlh2_lqo_grad_short}, we use 20 iterations with $c_1=10^{-4}$, $\rho=0.5$, and $\eta_{\mathrm{init}}=(1,1,1,1)$.
When TLH2 is used, its initial point is given by TLBT rather than a random stable system.

\tblack{Table~\ref{tab:experiments} reports the relative error for~\eqref{eq:opt_original} and the IMDb test accuracy without retraining.
Here, \textit{TestAcc.} denotes the classification accuracy on the IMDb test dataset.
The results show that a smaller relative error for~\eqref{eq:opt_original} generally leads to higher test accuracy.
In particular, although $r_{\mathrm{list}}=[16]\times4$ and $r_{\mathrm{list}}=[32,16,12,4]$ have the same number of trainable parameters, the latter allocates larger reduced dimensions to shallower layers according to our design principles, thereby achieving a smaller relative error and higher test accuracy.
Here, for $r_{\mathrm{list}}=[32,16,12,4]$, TLH2 generated an unstable reduced model and thus failed to construct a valid compressed model; this can occur because, unlike Alg.~\ref{alg:disc_cascade_tlh2_lqo_grad_short}, TLBT and TLH2 do not explicitly enforce the stability constraint \tred{when constructing the reduced model}.
\tred{For the two largest settings, Alg.~\ref{alg:disc_cascade_tlh2_lqo_grad_short} yields only marginal improvements over TLBT; its role is to provide a stability-constrained refinement from a feasible initializer rather than to guarantee a substantial accuracy gain in every case.}
Moreover, Alg.~\ref{alg:disc_cascade_tlh2_lqo_grad_short} with $r_{\mathrm{list}}=[96,48,36,12]$ achieves a relative error of $\tred{2.366\times 10^{-3}}$ and a test accuracy of $0.8669$, which is comparable to the pretrained model, while reducing the number of trainable parameters from 207,490 to 83,650, i.e., by approximately 60\%.
In addition, relative to the ``Text'' results in \cite[Fig.2]{gwak2024layer}, at a comparable reduced state dimension $r_{\text{list}}=[16]\times 4$, Algorithm~\ref{alg:disc_cascade_tlh2_lqo_grad_short}'s \textit{TestAcc.} is superior.
These results indicate that the proposed output-error bound is useful for designing layerwise reduced dimensions and constructing accurate compressed models without retraining.}

\begin{table}[t]
\centering
\caption{Compressed-model performance on IMDb without retraining.}
\label{tab:experiments}
\setlength{\tabcolsep}{6pt}
\begin{tabular}{@{\extracolsep{\fill}}cccc}
\toprule
$r_{\mathrm{list}}$ 
& Method 
& Relative error for~\eqref{eq:opt_original}
& \textit{TestAcc.} \\
\midrule
$[16]\times4$
& TLBT                 
& $6.330{\times}10^{-1}$ 
& $0.7615$ \\
& TLH2                 
& $6.101{\times}10^{-1}$ 
& $0.7642$ \\
& Alg.\ref{alg:disc_cascade_tlh2_lqo_grad_short} (TLBT init.) 
& $\tred{6.203{\times}10^{-1}}$ 
& $\tred{0.7621}$ \\
& Alg.\ref{alg:disc_cascade_tlh2_lqo_grad_short} (TLH2 init.) 
& $\tred{6.101{\times}10^{-1}}$ 
& $\tred{0.7642}$ \\
\midrule
$[32,16,12,4]$
& TLBT                 
& $3.147{\times}10^{-1}$ 
& $0.8213$ \\
& Alg.\ref{alg:disc_cascade_tlh2_lqo_grad_short} (TLBT init.) 
& $\tred{3.073{\times}10^{-1}}$ 
& $\tred{0.8138}$ \\
\midrule
$[64,32,24,8]$
& \tred{TLBT}
& $\tred{2.171{\times}10^{-2}}$
& $\tred{0.8627}$ \\
& Alg.\ref{alg:disc_cascade_tlh2_lqo_grad_short} (TLBT init.)
& $\tred{2.158{\times}10^{-2}}$ 
& $\tred{0.8628}$ \\
\midrule
$[96,48,36,12]$
& \tred{TLBT}
& $\tred{2.375{\times}10^{-3}}$
& $\tred{0.8668}$ \\
& Alg.\ref{alg:disc_cascade_tlh2_lqo_grad_short} (TLBT init.) 
& $\tred{2.366{\times}10^{-3}}$
& $0.8669$ \\
\bottomrule
\end{tabular}
\end{table}


\section{Summary}\label{sec:summary}
In this paper, we \tblack{have} derived an \tblack{upper bound on the overall output error} for two Deep SSMs incorporating LQO systems and constructed a gradient-based $h^2$ MOR \tblack{algorithm} that approximates the final intermediate output.
\tblack{In particular, through theoretical analysis and numerical experiments, we have demonstrated that assigning larger reduced state dimensions to shallow layers and performing MOR to reduce the derived upper bound is effective in reducing the overall output error.}

\tblack{As future work, we will evaluate the proposed method on physical and control-oriented tasks.}
\tblack{We will also investigate MOR methods for more general nonlinear systems, namely systems whose output equation in \eqref{eq:d-LQO-dssm} is not quadratic but given by a $K$th-order polynomial output map, and their application to model compression.}


\section*{Acknowledgment}
This work was supported by JSPS KAKENHI under Grant Numbers 23K28369, 25KJ0986, and 26K03232.

\appendices

\section{Detailed Derivation of Lemma~\ref{lem:discrete_lqo_gradients}}
\label{app:lemma2_derivation}

We adapt the perturbation and adjoint-equation argument of~\cite{zulfiqar2024time} to discrete-time complex LQO systems. The starting point is the Gramian expression~\eqref{eq:single_h2_error} for the Hermitian quadratic output $x^\ast M_i x$.

Define
\begin{align}
  \tilde H_L&:=\tilde Y_L+2\tilde Z_L,&
  \hat H_L&:=\hat Y_L+2\hat Z_L.
\end{align}
Using the definitions of $\tilde W_L$ and $\hat W_L$ in Lemma~\ref{lem:discrete_lqo_gradients}, equations~\eqref{eq:Ytil_t}--\eqref{eq:Zhat_t} give the adjoint Stein equations
\begin{align}
\tilde H_L
 &=A^\ast\tilde H_L\hat A+\tilde W_L
   -S_L^\ast\tilde W_L\hat S_L,
   \label{eq:adjoint_cross_arxiv}\\
\hat H_L
 &=\hat A^\ast\hat H_L\hat A+\hat W_L
   -\hat S_L^\ast\hat W_L\hat S_L.
   \label{eq:adjoint_reduced_arxiv}
\end{align}
The second equation has a Hermitian solution because $\hat W_L$ is Hermitian. We use the Stein trace identity: if
\begin{align*}
X=A^\ast X\hat A+F,
\qquad Q=AQ\hat A^\ast+G,
\end{align*}
then cyclic invariance of the trace gives
\begin{align}
\langle F,Q\rangle=\langle X,G\rangle.
\label{eq:stein_trace_identity}
\end{align}

Differentiating~\eqref{eq:single_h2_error} gives
\begin{align}
\mathrm{d}\varphi
={}&-2\langle\tilde W_L,\mathrm{d}\tilde P_L\rangle
 +\langle\hat W_L,\mathrm{d}\hat P_L\rangle
 +\mathrm{d}\varphi_{\mathrm{dir}},
\label{eq:first_variation_arxiv}
\end{align}
where $\mathrm{d}\varphi_{\mathrm{dir}}$ collects the explicit dependence on $\hat C$ and $\hat M_i$.  This term vanishes for variations of $\hat A$ and $\hat B$.

\subsection{Reduced State Matrix}
Let $E=\mathrm{d}\hat A$ be an arbitrary perturbation. The infinite-horizon Gramians~\eqref{eq:Ptil_inf} and~\eqref{eq:Phat_inf} satisfy
\begin{align}
\mathrm{d}\tilde P
 &=A\,\mathrm{d}\tilde P\,\hat A^\ast
   +A\tilde P E^\ast,
   \label{eq:dPtilde_A_arxiv}\\
\mathrm{d}\hat P
 &=\hat A\,\mathrm{d}\hat P\,\hat A^\ast
   +E\hat P\hat A^\ast+\hat A\hat P E^\ast.
   \label{eq:dPhat_A_arxiv}
\end{align}
The factor $A$ in~\eqref{eq:dPtilde_A_arxiv}, and analogously the factor $\hat A$ in~\eqref{eq:dPhat_A_arxiv}, is specific to the discrete-time Stein equation.  For comparison, differentiating the continuous-time cross Lyapunov equation gives a term $\tilde P E^\ast$ without the factor $A$.

The finite-horizon endpoint terms subtracted from $\tilde P$ and $\hat P$ are $S_L\tilde P\hat S_L^\ast$ and $\hat S_L\hat P\hat S_L^\ast$, respectively. Since $\hat S_L=\hat A^L$, the required Fr\'echet derivative is
\begin{align}
\mathcal T_{\hat A,L}(E)
 &:=\mathrm{d}(\hat A^L)[E]
   =\sum_{j=0}^{L-1}\hat A^jE\hat A^{L-1-j}.
   \label{eq:frechet_power_arxiv}
\end{align}
Its adjoint $\mathcal T^\ast_{\hat A,L}$ is defined in Lemma~\ref{lem:discrete_lqo_gradients}.
Thus, differentiating
$\tilde P_L=\tilde P-S_L\tilde P\hat S_L^\ast$ and
$\hat P_L=\hat P-\hat S_L\hat P\hat S_L^\ast$ yields
\begin{align}
\mathrm{d}\tilde P_L
={}&\mathrm{d}\tilde P
-S_L\mathrm{d}\tilde P\hat S_L^\ast
-S_L\tilde P\mathcal T_{\hat A,L}(E)^\ast,
\label{eq:dPtildeL_A_arxiv}\\
\mathrm{d}\hat P_L
={}&\mathrm{d}\hat P
-\mathcal T_{\hat A,L}(E)\hat P\hat S_L^\ast
-\hat S_L\mathrm{d}\hat P\hat S_L^\ast \nonumber\\
&-\hat S_L\hat P\mathcal T_{\hat A,L}(E)^\ast.
\label{eq:dPhatL_A_arxiv}
\end{align}
Substituting~\eqref{eq:dPtildeL_A_arxiv} and~\eqref{eq:dPhatL_A_arxiv} into~\eqref{eq:first_variation_arxiv}, the Stein trace identity~\eqref{eq:stein_trace_identity} and adjoint equations~\eqref{eq:adjoint_cross_arxiv}--\eqref{eq:adjoint_reduced_arxiv} eliminate the Gramian variations:
\begin{align}
\mathrm{d}\varphi
={}&-2\langle\tilde H_L,A\tilde P E^\ast\rangle
  +2\langle\hat H_L,\hat A\hat P E^\ast\rangle \nonumber\\
&+2\Bigl\langle\tilde W_L^\ast S_L\tilde P-\hat W_L\hat S_L\hat P,
  \mathcal T_{\hat A,L}(E)\Bigr\rangle.
\label{eq:dphi_A_stein_arxiv}
\end{align}
Using $\langle H,XE^\ast\rangle=\langle H^\ast X,E\rangle$ and the adjoint of $\mathcal T_{\hat A,L}$, we obtain
\begin{align}
\mathrm{d}\varphi
=2\big\langle
&-\tilde H_L^\ast A\tilde P
+\hat H_L\hat A\hat P \nonumber\\
&+\mathcal T^\ast_{\hat A,L}
 \bigl(\tilde W_L^\ast S_L\tilde P
       -\hat W_L\hat S_L\hat P\bigr),E
\big\rangle.
\label{eq:dphi_A_before_split_arxiv}
\end{align}
Finally, split $\tilde P=\tilde P_L+(\tilde P-\tilde P_L)$ and
$\hat P=\hat P_L+(\hat P-\hat P_L)$ in~\eqref{eq:dphi_A_before_split_arxiv}.  The terms involving $\tilde P_L$ and $\hat P_L$ give the first two terms of $\nabla_{\hat A}\varphi$, while the remaining terms are exactly $L_L$.  Hence
\begin{align}
\nabla_{\hat A}\varphi
=2\bigl(-\tilde H_L^\ast A\tilde P_L
 +\hat H_L\hat A\hat P_L+L_L\bigr),
\end{align}
which is the first formula in Lemma~\ref{lem:discrete_lqo_gradients}.

\subsection{Variations with respect to the remaining matrices}
Let $F=\mathrm{d}\hat B$ be an arbitrary perturbation. Since $S_L$ and $\hat S_L$ do not depend on $\hat B$, differentiating~\eqref{eq:Ptil_inf} and~\eqref{eq:Phat_inf} gives
\begin{align}
\mathrm{d}\tilde P
 &=A\,\mathrm{d}\tilde P\,\hat A^\ast+BF^\ast,\\
\mathrm{d}\hat P
 &=\hat A\,\mathrm{d}\hat P\,\hat A^\ast
   +F\hat B^\ast+\hat B F^\ast.
\end{align}
Applying the endpoint decompositions and the Stein trace identity to~\eqref{eq:first_variation_arxiv} gives
\begin{align}
\mathrm{d}\varphi
=2\langle-\tilde H_L^\ast B+\hat H_L\hat B,F\rangle.
\end{align}

The controllability Gramians do not depend on $\hat C$ or $\hat M_i$, so their derivatives follow directly from~\eqref{eq:single_h2_error}. For $G=\mathrm{d}\hat C$ and Hermitian $K_i=\mathrm{d}\hat M_i$, respectively,
\begin{align}
\mathrm{d}\varphi
 &=2\langle-C\tilde P_L+\hat C\hat P_L,G\rangle,\\
\mathrm{d}\varphi
 &=2\big\langle-\tilde P_L^\ast M_i\tilde P_L
       +\hat P_L\hat M_i\hat P_L,K_i\big\rangle.
\end{align}
These variations give the remaining gradients in Lemma~\ref{lem:discrete_lqo_gradients}.

\section{Proofs}\label{app:proofs}
\begin{proof}[Proof of Lemma~\ref{lem:LN-Lip}]
Let $P:=I-\frac{1}{m}\mathbf{1}\mathbf{1}^{\top}$, $c:=z-\mu\mathbf{1}= Pz$, and $D:=\mathrm{diag}(\gamma_1)$.
Then, for $\mathrm{LN}_{\gamma_1,\gamma_2}(z)= D\,\frac{c}{\sigma}+\gamma_2$, its Jacobian is 
$\nabla_z \mathrm{LN}=D\Bigl(\frac{I}{\sigma}-\frac{cc^\top}{\sigma^{3}m}\Bigr)P$.
Since $\|P\|_2=1$, $\|D\|_2=\|\gamma_1\|_{{\infty}}$, and $\bigl\|\tfrac{I}{\sigma}-\tfrac{cc^\top}{\sigma^{3}m}\bigr\|_2\le \tfrac{1}{\sigma}$ (because $\sigma^2\ge\varepsilon$), it follows for all $z$ that
$\|\nabla_z \mathrm{LN}\|_2\le \frac{\|\gamma_1\|_{{\infty}}}{\sigma}\le \frac{\|\gamma_1\|_{{\infty}}}{\sqrt{\varepsilon}}$,
which gives the upper bound.

For the lower bound, at $c=0$ we have $\sigma=\sqrt{\varepsilon}$ and $\nabla_z\mathrm{LN}=\tfrac{1}{\sqrt{\varepsilon}}DP$, hence 
$\operatorname{Lip}(\mathrm{LN})
\;\ge\;\frac{1}{\sqrt{\varepsilon}}\max_{\|x\|=1}\|DPx\|
\;\ge\;\frac{\|\gamma_1\|_{\infty}}{\sqrt{\varepsilon}}\sqrt{1-\frac{1}{m}}$.
\end{proof}

\begin{proof}[Proof of Theorem~\ref{thm:convergence-ssm}]
The proof follows~\cite[Thm. 2]{sakamoto2026data}.
The boundedness assumption makes $K:=\overline{\{\hat\theta_\ell\}}$ compact.
By construction, $f$ in~\eqref{eq:opt_original} is smooth \tred{near $K$ by the positive-error assumption}.
\tred{The finite-horizon $f$ is semialgebraic and hence has the KL property.}
Hence, the extreme value theorem implies that $\nabla f$ is $L_{\nabla}$--Lipschitz on $K$ \tred{for some $L_{\nabla}<\infty$}.

\tred{Let $\eta_{\max}:=\|\eta_{\mathrm{init}}\|_\infty$.}
For $s_\ell:=\hat\theta_{\ell+1}-\hat\theta_\ell$, we have
$\|s_\ell\|^2
=\eta_\Lambda^2\|\nabla_{\hat \Lambda} f\|^2+\eta_B^2\|\nabla_{\hat B} f\|^2
+\eta_C^2\|\nabla_{\hat C} f\|^2+\eta_U^2\|\nabla_{\hat U} f\|^2
\le \tred{\eta_{\max}}\; D_\ell$, and
$f(\hat\theta_{\ell+1}) \le f(\hat\theta_\ell)
-\frac{c_1}{\tred{\eta_{\max}}}\,\|s_\ell\|^2$.
Since $f$ is bounded below on $K$, we obtain 
$\sum_{\ell=0}^\infty \|s_\ell\|^2<\infty$, hence $\|s_\ell\|\to 0$.

Because the iterates remain at a positive distance from the instability boundary \tred{and the local Lipschitz bound gives a uniform Armijo threshold}, there exists $\underline\eta>0$ with \tred{$\min\{\eta_\Lambda,\eta_B,\eta_C,\eta_U\}\ge\underline\eta$} for all accepted steps. \tred{Hence,} $\|\nabla f(\hat\theta_{\ell+1})\|\le (L_\nabla+\underline\eta^{-1})\|s_\ell\|$.

The remaining argument follows~\cite[Thm. 2]{sakamoto2026data} and~\cite[Thm.~3.2]{attouch2013convergence}, proving that Alg.~\ref{alg:disc_cascade_tlh2_lqo_grad_short} converges to a stationary point in $K$.
\end{proof}    

\section{WikiText-103 Results}\label{app:wikitext103}
\tblack{We additionally evaluate the proposed method on WikiText-103~\cite{merity2016pointer} using perplexity (PPL; lower is better).
The pretrained model uses $n=256$, $m=256$, $c=1$, $L=1024$, and $\xi=8$, and achieves a test PPL of $42.12$.}

\tblack{Table~\ref{tab:wikitext103} reports results for Alg.~\ref{alg:disc_cascade_tlh2_lqo_grad_short} with TLBT initialization.
The bottom two rows have the same number of trainable parameters.
Reducing the relative error for~\eqref{eq:opt_original} improves PPL.
In particular, $r_{\mathrm{list}}=[192,176,160,160,160,160,144,128]$ allocates larger reduced dimensions to shallow layers and achieves a test PPL of 42.87, close to that of the pretrained model.
Large relative errors sharply degrade PPL through error accumulation in next-token prediction.
These results suggest that the proposed bound is also useful for compressed-model design.}

\begin{table}[t]
\centering
\caption{Compressed-model performance on WikiText-103 without retraining.}
\label{tab:wikitext103}
\scriptsize
\setlength{\tabcolsep}{4pt}
\begin{tabular}{ccc}
\toprule
$r_{\mathrm{list}}$ 
& Relative error for~\eqref{eq:opt_original}
& \textit{TestPPL.} \\
\midrule
$[96]\times8$
& $5.377{\times}10^{-1}$
& $\tred{111.93}$ \\
\midrule
$[160]\times8$
& $2.336{\times}10^{-1}$
& $43.83$ \\
\midrule
\makecell{$[192,176,160,160,160,160,144,128]$}
& $\mathbf{\tred{9.833{\times}10^{-2}}}$
& $\mathbf{42.87}$ \\
\bottomrule
\end{tabular}
\end{table}

\section{Normalized Bound and Empirical Output Errors}
\label{app:normalized_bound}
The upper bound can be highly conservative in magnitude. To compare relative trends, Fig.~\ref{fig:normalized_bound_arxiv} shows an auxiliary IMDb evaluation of reduced models obtained with TLBT initialization. For each test input, the empirical error is $e_\xi=\|s_{\mathrm{out}}-\hat{s}_{\mathrm{out}}\|_{\ell_L^\infty}$. Mean error and max error are its mean and maximum over the test set, respectively; upper bound is the right-hand side of Corollary~\ref{cor:constant_bound}. Each quantity is divided by its value at $r_{\mathrm{list}}=[4,12,16,32]$.

\begin{figure}[H]
\centering
\resizebox{\columnwidth}{!}{%
\begin{tikzpicture}[x=1.20cm,y=1.05cm]
  \draw[->] (0.65,-2.85) -- (5.65,-2.85);
  \draw[->] (0.80,-2.95) -- (0.80,0.25);
  \node[rotate=90,font=\small] at (-0.30,-1.35) {Normalized value};
  \foreach \y/\lab in {0/$10^{0}$,-1/$10^{-1}$,-2/$10^{-2}$} {
    \draw[gray!35] (0.80,\y) -- (5.20,\y);
    \draw (0.75,\y) -- (0.85,\y) node[left=3pt] {\scriptsize \lab};
  }
  \foreach \x/\lab in {
    1/{[4,12,16,32]},
    2/{[16]\!\times\!4},
    3/{[32,16,12,4]},
    4/{[64,32,24,8]},
    5/{[96,48,36,12]}} {
    \draw (\x,-2.80) -- (\x,-2.90);
    \node[rotate=28,anchor=east,font=\scriptsize] at (\x,-3.00) {$\lab$};
  }
  \draw[blue,thick,mark=*] plot coordinates {(1,0) (2,-0.109) (3,-0.318) (4,-1.208) (5,-1.633)};
  \draw[orange!85!black,thick,dashed,mark=square*] plot coordinates {(1,0) (2,-0.043) (3,-0.244) (4,-1.009) (5,-1.500)};
  \draw[red!75!black,thick,dash dot,mark=triangle*] plot coordinates {(1,0) (2,-0.167) (3,-0.473) (4,-1.626) (5,-2.586)};
  \draw[blue,thin] (5,-1.633) -- (5.15,-1.78) node[right,font=\scriptsize] {Mean error};
  \draw[orange!85!black,thin] (5,-1.500) -- (5.15,-1.38) node[right,font=\scriptsize] {Max error};
  \draw[red!75!black,thin] (5,-2.586) -- (5.15,-2.586) node[right,font=\scriptsize] {Upper bound};
\end{tikzpicture}%
}
\caption{Normalized empirical output errors and upper bound in an auxiliary IMDb evaluation. Each series is normalized by its value for $r_{\mathrm{list}}=[4,12,16,32]$.}
\label{fig:normalized_bound_arxiv}
\end{figure}

The bound decreases more rapidly than the empirical errors, so this comparison does not establish tightness or accurate prediction of absolute error. Nevertheless, all three quantities decrease across these configurations, supporting the use of the bound to guide layerwise order allocation.

\bibliographystyle{IEEEtran}
\bibliography{main.bib}
\end{document}